\documentclass[twocolumn,english]{IEEEtran}
\usepackage{graphicx}
\usepackage[normalem]{ulem}
\usepackage{epstopdf}
\usepackage{flushend}
\usepackage{cite}
\usepackage{url}
\usepackage{bm}
\usepackage{booktabs}
\usepackage{enumerate}
\usepackage[T1]{fontenc}
\usepackage{amstext}
\usepackage{color}
\usepackage{array}
\usepackage{amsthm}
\usepackage{amsfonts}
\usepackage{amsmath}
\usepackage{amssymb}
\newtheorem{theorem}{Theorem}
\newtheorem{definition}{Definition}
\newtheorem{proposition}{Proposition}

\newtheorem{lemma}{Lemma}
\begin{document}
\title{Efficient Multiplex for Band-Limited Channels: Galois-Field Division Multiple Access }

\author{H.~M. de Oliveira, R.~M. Campello de Souza and A.~N. Kauffman \\
Authors are with the Federal University of Pernambuco, \\
Recife, PE, Brazil, Email: hmo@de.ufpe.br, ricardo@ufpe.br}
\maketitle

\begin{abstract}
A new Efficient-bandwidth code-division-multiple-access (CDMA) for band-limited channels is introduced  which is based on finite field transforms. A multilevel code division multiplex exploits orthogonality properties of nonbinary sequences defined over a complex finite field. Galois-Fourier transforms contain some redundancy and just cyclotomic coefficients are needed to be transmitted yielding compact spectrum requirements. The primary advantage of such schemes regarding classical multiplex is their better spectral efficiency. This paper estimates the \textit{bandwidth compactness factor} relatively to Time Division Multiple Access TDMA showing that it strongly depends on the alphabet extension. These multiplex schemes termed Galois Division Multiplex (GDM) are based on transforms for which there exists fast algorithms. They are also convenient from the implementation viewpoint since they can be implemented by a Digital Signal Processor.
\end{abstract}

\begin{IEEEkeywords}
Digital multiplex, Code-division multiple access, Hartley-Galois transform, Finite field transforms, Spread sequence design.
\end{IEEEkeywords}

\section{Introduction}
\label{sec:introduction}
The main title of this paper is, apart from the term multiplex, literally identical to a Forney, Gallager and co-workers paper issued more than one decade ago [FOR et al. 84], which analyzed the benefits of coded-modulation techniques. The large success achieved by Ungerboeck's coded-modulation came from the way of introducing redundancy in the encoder [UNG 82]. In classical channel coding, redundant signals are appended to information symbols in a way somewhat analogous to time division multiplex TDM (envelope interleaving). It was believed that introducing error-control ability would increase bandwidth. An efficient way of introducing such an ability without sacrificing rate nor requiring more bandwidth consists in adding redundancy by  an alphabet expansion. This technique is particularly suitable for channel in the narrow-band region.  A similar reasoning occurs in the multiplex  framework  where  many people nowadays believe that mux must increase bandwidth requirements.
\\
\\
One of the most powerful tools in Communications is the Fourier transform, specially its discrete version DFT. On the other side, applications of finite field have a renewed explosion of interest in the last decades. A discrete Fourier transform for finite fields (FFFT) was introduced by Pollard [POL 71]. It has successfully applied to perform discrete convolution and as a tool of image processing [REE et al. 77, REE-TRU 79] among many other applications. Besides classical FFFT we are concerned with a new finite field discrete version [CAM et al. 98] of the integral transform introduced by R.V.L. Hartley [HAR 42, BRI 92]. Alike classical Galois-Fourier transforms [BLA 79], Finite Field Hartley transform GHT [CAM et al. 98] defined on a Gaussian integer set GI($p^m$) contains some redundancy and only the cyclotomic coset leaders of transform coefficients are needed to be transmitted. This yields a new \textbf{Efficient-bandwidth Code Division multiplex for band-limited channels}. These mux may present lower bandwidth requirements than TDM by using an argument somewhat analogous to coded-modulation. Tradeoffs between extension of the alphabet and bandwidth are exploited in the sequel.
\\
\\
In the present work, the coded-modulation idea is adapted to multiplex: Information streaming from users are not combined by interleaving (like TDM) but rather by signal alphabet expansion. The mux of users' sources over a Galois Field GF($p$) deals with an expanded signal set having symbols from an extension field GF($p^m$), $m>1$. As a consequence, the multiplex of $N$ band limited channels of identical maximal frequency $B$ leads to  \textit{bandwidth requirements less than $N.B$}, in contrast with TDMed or FDMed signals.
\\
\\
The design of such an \textit{efficient bandwidth mux} is based upon the finite field structure, specifically, it consists of applying finite-field Fourier transforms (FFFT). This paper shows that the 'bandwidth compactness factor' relatively to TDM depends on the length $m$, the alphabet extension. 
\\
\\
Another point to mention is that the superiority of digital mux regarding analogic mux is essentially due to the low complexity of TDM. The majority of today multiplex systems follows plesiocrhonous (PDH) or synchronous (SDH) Hierarchy. Besides presenting higher spectral efficiency (bits/s/Hz) than classical mux, the new CDM schemes here introduced are based on fast transforms so they also seem to be attractive from the implementation viewpoint. Although most mux systems today be intended to optical fiber which are not yet bandlimited channels, multiplex has also been adopted on satellite channels. Probably applications of such multiplex will be on transponders and specially on cellular communications. Cable Television (CATV) can also have benefit by using such a technique.

\section{A New Mux Scheme: Galois-Division-Multiplex}
\label{sec:mux}

Digital Multiplex normally alludes Time Division Multiplex (TDM). However, it also can be achieved by  Coding Division Multiplex (CDM). The CDM  has recently been focus of interest, specially after the IS-154 standardization of the CDMA system for cellular telephone. In this section we introduce a new class of mux schemes based upon finite-field transforms which can be implemented by  fast transform algorithms. \textit{Classical multiplex increases simultaneously the transmission rate and the bandwidth by the same factor}, keeping thus the spectral efficiency unchanged. In order to achieve (slight) better spectral efficiencies, classical CDMA uses waveforms presenting a non zero but residual correlation. 
\\
\\
Given a signal $v$ over a finite field GF($p$) we deal with the Galois domain considering the spectrum $V$ over an extension field GF($p^m$) which corresponds to the Finite Field Transform (Galois Transform) [BLA 79, CAM et al. 98].
\\
\\
As an alternative and attractive implementation, the multiplex is carried out by a Finite Field Transform (FFFT/FFHT) and the DEMUX corresponds exactly to a Inverse Finite-Field Transform  of length $N | p^m-1$.

\begin{figure}
\centering
{\includegraphics[width=0.8\columnwidth]{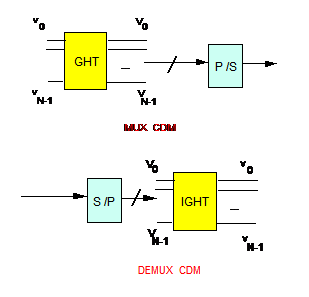}}
\caption{Implementation of Galois-Based mux or multiple-access.}
\label{fig:implementation}
\end{figure}
Each symbol on the ground field GF($p$) have duration $T$ seconds. An $N$-user mux can be designed on the extension field GF($p^m$) where $N | p^m-1$. For the sake of simplicity, we begin with $m$=1 and consider a ($p-1$)-channel mux as follows. Typically, we can consider GF(3) corresponding to Alternate Mark Inversion AMI signaling.

\begin{definition}  A Galois modulator carries a pairwise multiplication between a signal $(v_0,v_1,...,v_{N-1}),~v_i \in GF(p)$ and a carrier $(c_0,c_1,...,c_{N-1}),~with~c_i \in GI(p)$.
\end{definition}

\begin{figure}
Representation.
\centering{\includegraphics[width=0.45\columnwidth]{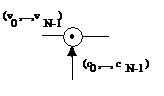}}
\end{figure}

A ($p-1$)-CD consider digital carrier sequences per channel as versions of cas function  over the Galois (complex) field GI($p$). The cas (cos and sin) function is defined in terms of finite field trigonometric functions [CAM et al. 98], $cas_i k:=cos_i k+ sin_i k$.\\
\\
Carrier  0: \\   
$  ~ ~ ~ ~ \{ cas_0 0 ~cas_0 1 ~cas_0 2~   ... cas_0 (N-1) \}$ \\
Carrier 1:\\
$ ~ ~ ~ ~ \{ cas_1 0 ~cas_1 1 ~cas_1 2~   ... cas_1 (N-1) \}$ \\
...\\
Carrier $j$:\\
$~ ~ ~ ~ \{ cas_j 0 ~cas_j 1 ~cas_j 2~  ... cas_j (N-1) \}$ \\
...\\
Carrier $N-1=p-1$:\\
$~ ~ ~ ~ \{ cas_{N-1} 0 ~cas_{N-1} 1 ~cas_{N-1} 2~  ... cas_{N-1} (N-1) \}$.\\
\\
The cyclic digital carrier has the same duration $T$ of a input modulation symbol, so that it carries $N$ slots per data symbol. The interval of each cas-symbol is $T/N$ and therefore the bandwidth expansion factor by multiplexing $N$ channels may be roughly $N$, the same result as FDM and TDM/PAM. 
\\
A first scheme of the multiplex is showed: The output corresponds exactly to the Galois-Hartley Transform of the "users"-vector $(v_0,v_1,...,v_{N-1})$. 
\\
Therefore, it contains all the information of the channels. Each coefficient $V_k$ of the spectrum has duration $T/N$.
\begin{figure}
\centering{\includegraphics[width=0.9\columnwidth]{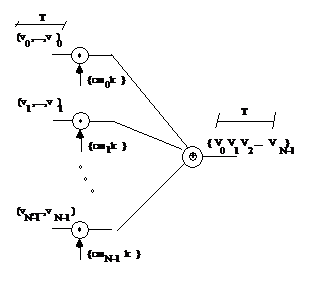}}
\caption{Galois-Field MUX: Spreading sequences.}
\label{fig:spread}
\end{figure}

These carriers can also be viewed as spreading waveforms. An $N$-user mux has $N$ 'spreading' sequences, one per channel. The requirements to achieve Welch's lower bound according to Massey and Mittelholzer [MAS-MIT 91] are hold by $ \left \{ cas_i k \right \}_{i=1}^{p-1}$ sequences. The matrix $\left [ \left \{ cas_i k \right \} \right ]$ presents both orthogonal rows and columns and have the same 'energy'.
\\
\\
A naive example is present in order to illustrate such an approach. A 4-channel mux over GF(5) can be easily implemented $i=0,1,...,p-2=3$. It is straightforward to see that such signal are not FDMed nor TDMed.

\begin{table}[h]
\centering
\caption{CAS FUNCTION ON GI(5), $\alpha=2$ element of order 4.} \label{tab1}
\begin{tabular}{c c c c}
\hline
$cas_0 0=1+j0$  & $cas_0 1=1$  & $cas_0 2=1$ & $cas_0 3=1$\\
$cas_1 0=1+j0$  & $cas_1 1=3j$  & $cas_1 2=4$ & $cas_1 3=2j$\\
$cas_2 0=1+j0$  & $cas_2 1=4$  & $cas_2 2=1$ & $cas_2 3=4$\\
$cas_3 0=1+j0$  & $cas_3 1=2j$  & $cas_3 2=4$ & $cas_3 3=3j$\\
\hline
\end{tabular}
\end{table}

\begin{figure}
\centering{\includegraphics[width=0.9\columnwidth]{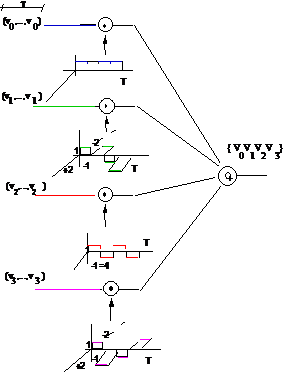}}
\caption{Interpreting Hartley-Galois Transform over GF(5) as spreading waveforms.}
\label{fig:interpret}
\end{figure}

The digital carriers are defined on a Galois field GI($p$) and consider the element $\sqrt{-1}$ which may or not belongs to GF($p$) although the original definition [CAM et al 98] consider -1 as a quadratic non residue in GF($p$). Two distinct cases are to be considered: $p=4k+1$ or $p=4k+3$, $k$ integer. If $p \equiv 1$ mod 4, then -1 is a quadratic residue. For instance, considering $j \in GF(5)$ then $2^2 \equiv -1~mod~5$ so $j= \sqrt{-1} \equiv 2~mod~5$. Two-dimensional digital $\{ cas_i k \}_{k=0}^{p-1}$ carriers degenerated to one-dimensional carriers.
\\
\\
Considering the above example, carriers are reduced to Walsh carrier!
\\
\\
$\{cas_0k\}= \{1,  1, 1, 1\} = \{1,1, 1, 1\}$\\
$\{cas_1k\}= \{1,  1, 4, 4\} = \{1,1,-1,-1\}$\\
$\{cas_2k\}= \{1,  4, 1, 4\} = \{1,-1,1, -1\}$\\
$\{cas_3k\}= \{1,  4, 4, 1\} = \{1,-1,-1, 1\}$\\
\\
\begin{equation}
\begin{bmatrix}
 1&  1&  1& 1\\ 
 1&  1&  -1& -1\\ 
 1&  -1&  1& -1\\ 
 1&  -1&  -1& -1
\end{bmatrix}
\Leftrightarrow 
\begin{bmatrix}
 1&  1&  1& 1\\ 
 1&  -1&  1& -1\\ 
 1&  1&  -1& -1\\ 
 1&  -1&  -1& -1
\end{bmatrix}
=
[WAL(k,i)].
\end{equation}

In the absence of noise, there is no cross-talk from any user to any other one which corresponds to orthogonal carrier case.\\

If channels number 1,2,3, and 4 are transmitting $\{4, 0, 1, 2\}$,  respectively, the mux output will be $(2, 3+4j, 3, 3+j)$, which corresponds to\\
\{4,0,1,2\} $\otimes$ \{1,1,1,1\}   $\equiv$ 2  mod 5\\
\{4,0,1,2\} $\otimes$ \{1,j3,4,j2\} $\equiv$ 3+4j  mod 5\\
\{4,0,1,2\} $\otimes$ \{1,4,1,4\}   $\equiv$ 3  mod 5\\
\{4,0,1,2\} $\otimes$ \{1,j2,4,j3\} $\equiv$ 3+j  mod 5\\
\\
There is no gain when the transform is taken without alphabet extension. However, we have a nice interpretation of CDM based on finite field transforms.

\section{ A New CDM scheme based on Hartley-Galois Transforms}
\label{sec:CDM}

So far we have considered essentially GHT on the ground field GF($p$), i.e., Finite Field Transforms from G($p$) to G($p$). Extension fields can be used and results are much more interesting GDMA- Galois-Field Division Multiple Access schemes. The advantage of the new scheme termed GDMA on FDMA / TDMA is regarding their higher spectral efficiency. 
\\
\\
The new multiplex is carried out over the \textit{Galois domain} instead the \textit{Frequency or time domain}. Figure 4 exhibit a block diagram of transform-based multiplexes. First, the Galois spectrum of $N$-user GF($p$) signals is evaluated. The spectral compression is achieved by eliminating the redundancy: only the leader of cyclotomic cosets are transmitted. The demultiplex is carried out (after signal regeneration) first recovering the complete spectrum by "filling" missing components from the received coset leaders. Then, the inverse finite field transform is computed so as to obtain the demux signals. Another additional feature is that GDM implementation can be very attractive using fast algorithms. 
\\
Suppose that users' data are $p$-ary symbols transmitted at a speed $B_1:=1/T$ bauds. Let us consider the problem of multiplexing $N$ users. Traditionally the bandwidth requirements will increase proportionally with the number $N$ of channels, i.e., $B_N=N.B_1$ Hz. 

\begin{figure}
\centering{\includegraphics[width=0.85\columnwidth]{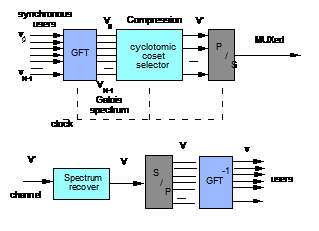}}
\caption{Multiplex based on Finite Field Transforms. Suppose that users' data are $p$-ary symbols transmitted at a speed $B-1:=1/T$ bauds. Let us consider the problem of multiplexing $N$ users. Traditionally the bandwidth requirements will increase proportionally with the number $N$ of channels, i.e., $B_N=N.B_1$ Hz.}
\label{fig:mux ffht}
\end{figure}

Heretofore the number of cyclotomic cosets using Galois-Fourier and Hartley-Fourier finite field spectrum is denoted $\nu_G$ and $\nu_H$, respectively. The clock driving GHT symbols is $N/\nu$ times faster than the input baud rate.\\

\begin{definition} The bandwidth compactness factor parameter $\gamma_{cc}$ is defined as $\gamma_{cc} :=N/\nu$. 
It plays a role somewhat similar than the coding asymptotic gain $\gamma_{cc}$ on coded modulation [UNG 82].            
\end{definition}

Transform-multiplex, i.e., mux based on finite field transforms are very attractive compared with FDM, TDM due their better spectral efficiency as it can be seem in Table II (appendix).
\\
Another point to mention is that instead compressing spectra (eliminating redundancy), all the coefficients can be used to introduce some error-correction ability. The valid spectrum sequences generate some sort of multilevel block code.
\\
\begin{lemma} For an $N$-user GDMA system over $GI(p^m)$ with $N~|~p^m-1$, only a number $\nu = \gamma_{cc}^{-1}N$  of finite-field transform coefficients required to be transmitted.
\end{lemma}
\begin{proof}
 According to M\"obius' inversion formula, $I_k(q)=\frac{1}{k}\sum_{d|k}^{~} \mu(d)q^{k/d}$  gives the number of distinct primitive polynomials of degree $k$ over GF($q$),  where $\mu$ is the Moebius function. Therefore, the number $\nu_g$ of cyclotomic sets on the Galois-Fourier transform  $(V_0,V_1,...,V_{N-1})$ is given by

\begin{equation}
\nu_g=\sum_{m}^{~}I_m(p)-1.
\end{equation}

Since each pair of cosets containing reciprocal roots is clustered,
\begin{equation}
\nu_H=\frac{\nu_G+(N~mod~2)}{2}+1. 
\end{equation}

\end{proof}

As a role-of-thumb, the number of cosets (in fact the gain $\gamma_{cc}$) when $N=p^m-1$ is roughly given by $\nu_G  \approx \left \lceil \frac{N}{m} \right \rceil $ and $\nu_H \approx \left \lceil {\frac{1}{2}  \left \lceil {\frac{N}{m}} \right \rceil +1} \right \rceil $, where  $\left \lceil x \right \rceil $ is the ceiling function (the smallest integer greater than or equal to $x$).
\\
\\
A simple example is presented below over $GF(3) \rightarrow  GI(3^3)$: Factoring $x^{26}-1$ one obtains $\nu_G=10$ and $\nu_H= 6$. For Hartley FFHT,  $V_k^{3}=V_{26-k3}$ (indexes modulo 26), according to Lemma 1 [CAM et al. 98].\\

\begin{table}[h]
\centering
\begin{tabular}{|l|l|}
FFFT cosets	   & FFHT cosets \\
C0=(0)	   &C0=(0)\\
C1=(1,3,9)	   &C1=(1,23,9,25,3,17)\\
C2=(2,6,18)	   &C2=(2,6,18,8,24,20)\\
C4=(4,12,10)  &C4=(4,14,10,22,12,16)\\
C5=(5,15,19)   &C5=(5,11,19,21,15,7)\\
C7=(7,21,11)   &C13=(13).\\
C8=(8,24,20)   &\\
C13=(13)            &\\
C14=(14,16,22) &\\
C17=(17,25,23).&
\end{tabular}
\end{table}

Another interesting possibility is multiplexing without the cyclotomic coset compression. Although such a GDM present the same spectral efficiency as TDM or FDM it introduces some error-correcting ability yielding a better performance. By way of interpretation, Hartley transform can be viewed as a kind of 'Digital Single Side Band' since the number of cyclotomic cosets of FFHT is roughly a half of that one of FFFT. 
\\We can say that '\textit{GDM/FFFT is to FDM/AM as GDM/FFHT is to FDM/SSB}.'

\begin{proposition} Gain of GDM. The gain on the number of channels GDMed regarding to TDM/FDM on the same bandwidth is $N-\nu$, which corresponds to $g\%=100(1-\gamma_{cc}^{-1} )$
\end{proposition}

\begin{proof} 
The bandwidth gain is $g_{band}=B_{TDM}/B_{GDM} =\gamma_{cc}$ and the saved Bandwidth is given by:
$B_{TDM} - B_{GDM}$. Calculating how many additional $B_1$-channels (users) can be introduced:
\begin{equation}
\frac {(B_{TDM} - B_{GDM})}{B_1}=(1- \frac {1}{\gamma_{cc}})N.
\end{equation}  
\end{proof}
Indeed, a more formal treatment of spectra should be tried. Power Spectral density calculations should be evaluated by using cyclic Autocorrelation Functions (ACF) of the carriers. As usual, it is assumed that users' data sequences are independent.

\begin{lemma} Users' signal sequences can be viewed as wide-sense stationary random processes in both time and Galois-domain, by assuming an uniform probability distribution on GF($p$) symbols.
\end{lemma}

\begin{proof}
time domain data stream:
\begin{eqnarray}
&...\left ( v_0^{(-1)},v_1^{(-1)},...,v_{N-1}^{(-1)} \right )~\left( v_0^{(0)},v_1^{(0)},...,v_{N-1}^{(0)} \right )~ \\ \nonumber
&\left( v_0^{(1)},v_1^{(1)},...,v_{N-1}^{(1)} \right )~\left( v_0^{(2)},v_1^{(2)},...,v_{N-1}^{(2)} \right ) ...
\end{eqnarray}
Supposing $v_i \in \left \{ 0,\pm 1,\pm 2,...,\pm \frac{(p-1)}{2} ) \right \}$ equally likely, the mean and the ACF of the discrete process are given respectively by ($E$ denotes the expected value):\\
$E \left \{ v_i^{(m)} \right \} =E \left \{ v_i \right \} =0 ~~(\forall i)$  and
\begin{equation}
R_{v} (j)=E \left \{ v_i^{(m)}.v_{i-j}^{(m)^*} \right \}=E \left \{ v_i.v_{i-j}^* \right \}=0 ~i\neq j,
\end{equation}
\begin{equation}
R_{v} (0)=0^2+1^2+2^2+3^2+...+(\frac{p-1}{2})^2:=\textbf{P}.
\end{equation}

Galois domain data stream: 
\begin{eqnarray}
&...\left ( V_0^{(-1)},V_1^{(-1)},...,V_{N-1}^{(-1)} \right )~\left( V_0^{(0)},V_1^{(0)},...,V_{N-1}^{(0)} \right )~ \\ \nonumber
&\left( V_0^{(1)},V_1^{(1)},...,V_{N-1}^{(1)} \right )~\left( V_0^{(2)},V_1^{(2)},...,V_{N-1}^{(2)} \right )...
\end{eqnarray}
From $V_k=\sum_{i=0}^{N-1} v_i cas_k i~(\forall k)$, it follows that:
\begin{equation}
E \left \{ V_k^{(m)} \right \} =E \left \{ V_k \right \}=\sum_{i=0}^{N-1}E(v_i).cas_k i=0 ~(\forall k).
\end{equation}
The ACF of the random process (Galois spectrum sequence) is:\\
$R_V(j)=E\left ( V_k V_{k-j}^{*} \right )$ that is\\
$R_V(j)=\sum_{i=0}^{N-1}\sum_{\Delta i=0}^{N-1}E\left ( V_i V_{i-\Delta i}^{*} \right )cas_k i.cas_{k-j}^{*}(i-\Delta i)$.\\
$R_V(j)=\sum_{i=0}^{N-1}\sum_{\Delta i=0}^{N-1} R_V (\Delta i) cas_k(i).cas_{k-j}^{*}(i-\Delta i)$.\\

Therefore, 
\begin{equation}
R_V(j)=R_V(0)\sum_{i=0}^{N-1}cas_k i.cas_{k-j}^{*}i.
\end{equation}
From orthogonality properties of $cas_k$-function [CAM et al. 98], it follows that
\begin{equation}
R_V(j)=0~~j\neq 0~and~R_V(0)=R_v(0).
\end{equation}     
\end{proof}
This result can be used to show that the multiplex based upon finite field Hartley transform does not shape the signal power spectrum. Denote by $S_b(t)$ the  complex  envelope of the generalized QAM (G-QAM) signal $s(t)$, i.e.,
$s(t)=\Re e \left \{ s_b(t)exp(j2\pi f_c t) \right \}$ where
\begin{eqnarray}
s_b(t)=\sum_{m=-\infty }^{+\infty }\sum_{k=0}^{N-1} V_k^{(m)}u_{k,m}(t).\\ \nonumber
u_{k,m}:=u(t-kT-mNT).
\end{eqnarray}
The corresponding 2-dimensional constellation is drawn in Figure 5 (see appendix).

\begin{theorem}
The Power Spectral Density of GDM signals is given by $|U(f)|^2$ which depends just on the shape of the spectrum of the $u(t)$.
\end{theorem}
\begin{proof}
 It can be shown that $s_b(t)$ is a cyclostationary process [GAD-FRA 75] but we treat it as a wide-sense stationary one by introducing a random phase uniformly distributed over one block. Therefore we consider a related signal
\begin{equation}
 \widetilde{s}_b(t)=\sum_{m=-\infty }^{+\infty }\sum_{k=0}^{N-1}V_k^{(m)}u_{k,m}(t;\theta )
\end{equation}
where $u_{k,m}(t;\theta)=u_{k,m}(t-\theta)$,  $\theta$ being uniformly distributed in the interval $(0,NT)$. The autocorrelation function (ACF)  of the complex envelope is
\begin{eqnarray}
R_{\widetilde{s}_b}(t,t-\tau )= ~~~~~~~~~~~~~~~~~~~~~~~~~~~~~~~~~~~~~~~~~~~~~~~& \\ \nonumber
~~~=\sum_{m=-\infty }^{+\infty }\sum_{\Delta m=-\infty }^{+\infty }\sum_{k=0}^{N-1}\sum_{\Delta k=0}^{N-1}E~V_k^{(m)}V_{k-\Delta k}^{(m-\Delta m)^*}&\\ \nonumber
\frac{1}{NT}\int_{0}^{NT}u_{k,m}(t,\theta )u_{k-\Delta k,m-\Delta m}^{*}(t-\tau ;\theta )d\theta&
\end{eqnarray}

where indexes $k-\Delta k$ are taken mod $N$.
\\
\\
We remark first that spectra (blocks) are transmitted independently so that complex symbols on different $N$-vectors are uncorrelated, yielding \\
\begin{equation}
\sum_{\Delta k=0}^{N-1}E~V_k^{(m)}V_{k-\Delta k}^{(m-\Delta m)*}=0~~\Delta m\neq 0.
\end{equation}
Furthermore, we suppose the sequence of $N$-dimensional signals is wide-sense stationary (Lemma 2) and that there exists a block ACF $R_V(j)$, $j=0,1,2,...N-1$. The ACF of the complex envelope is therefore
\\
\begin{eqnarray}
& R_{\widetilde{s}_b}(t,t-\tau )= \frac{1}{NT} \sum_{k=0}^{N-1} \sum_{j=0,j\neq k}^{N-1}R_V(j)\sum_{m=-\infty }^{+\infty }\\ \nonumber
&\int_{0}^{NT}u_{k,m}(t;\theta )u_{k-j,m}^{*}(t-\tau ;\theta )d\theta
\end{eqnarray}

By an appropriate change of variables, we obtain\\
$R_{\widetilde{s}_b}(t,t-\tau )=$\\
$=\frac{1}{NT}\sum_{k=0}^{N-1}\sum_{j=0,j\neq k}^{N-1}R_V(j)\int_{0}^{NT}u(\alpha +\tau -kT)u^{*}(\alpha)d\alpha $
Putting the above equation in a simpler notation results in:
\begin{equation}
R_{\widetilde{s}_b}(\tau )=R_{\widetilde{s}_b}(t,t-\tau )=\frac{1}{NT}\sum_{j=0}^{N-1}R_V(j)\vartheta (\tau -jT),
\end{equation}\\
where $\vartheta (\tau -jT)=\int_{-\infty }^{+\infty }u(\alpha +\tau -jT)u^{*}(\alpha)d\alpha$. \\
Indeed, $R_{\widetilde{s}_b}(\tau )=\frac{1}{NT}R_V(0)\vartheta (\tau)$. 
Taking now the Fourier transform of the ACF, we have finally by the Wiener-Kinchine relation
$S_{\widetilde{s}_b}(f )=\frac{\textbf{P}}{NT}\left | U(f) \right |^2$, so the power spectrum of the multiplexed signal follows directly from the modulation theorem.     
\end{proof}

What can we say about the alphabet extension? The greatest extension that can be used depends on the signal-to-noise ratio, since the total rate cannot exceed Shannon Capacity over the Gaussian channel.
$B_{GDM} \gamma_{cc} . log_{2} p \leq B_{GDM}.log_{2} \left ( 1+ \frac{S}{N} \right )$ bps,or
\begin{equation} 
\gamma_{cc} \leq log_{p}\left ( 1+\frac{S}{N} \right ).
\end{equation}
\section{conclusions}
\label{sec:conclusion}

New digital multiplex schemes based on finite field transforms have been introduced which are multilevel Code Division Multiplex. They are attractive owing to their better spectral efficiency regarding classical TDM/CDM which require a bandwidth expansion roughly proportional to the number of channels to be multiplexed. Moreover, the Galois-Field Division (GDM) implementation can be easily carried out by a Digital Signal Processor (DSP). Another nice payoff of GDM is that when Hartley Finite Field transform is used, the mux and demux hardware are exactly the same. It is proved that GDM based on Finite Field Hartley Transform does not shape the signal Power Spectrum. They can directly be applied in multiple access digital schemes. 

\section{Acknowledgments}
\label{sec:ack}
This first author express his deep indebtedness to Professor G\'erard Battail (ENST Paris) whose philosophy had a decisive influence on the way of looking coding and multiplex.\\

\textbf{appendix}\\

\begin{table}[b]
\caption{A spectral efficiency comparison for multiplex Systems.} \label{tab2}
\begin{tabular}{ c c c c }
	 &one-user & $N$-users TDMed & GDMed \\
\hline
Tx rate& $R_{i-user}=\frac {log_2 p}{T}$ & $\sum_{i}{} R_{i-user}$           & $ \sum_{i}{} R_{i-user}$ \\ 
                  & $=N \frac {log_2 p}{T}~bps$      & $= N \frac {log_2 p}{T}~bps$& \\
~& ~& ~&~\\
Bandwidth requirements & $B_1=\frac {1}{T}~Hz$ & $B_N= \frac {1}{T/ N}=NB_1~Hz$ & $B_{GDM}=\frac{1}{\gamma_{cc}}.(NB_1)~Hz$\\
~&~ & ~&~\\
Spectral efficiency & $\eta_{i-user}=log_{2} p$ bits/s/Hz & $\eta_{mux}=log_{2} p$ bits/s/Hz & $\eta_{GDM}= \gamma_{cc} log_{2} p$ bits/s/Hz\\
\hline
\end{tabular}
\end{table}

\newpage
\begin{minipage}[b]{1\linewidth}
\begin{IEEEbiography}[{\includegraphics[width=1in,height=1.25in,clip,keepaspectratio]{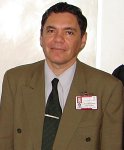}}]
{H\'elio Magalh\~{a}es de Oliveira} was born in Arcoverde,
Pernambuco, Brazil, in 1959. He received both the B.Sc. and M.S.E.E.
degrees in electrical engineering from Universidade Federal de
Pernambuco (UFPE), Recife, Pernambuco, in 1989 and 1983,
respectively. Then he joined the staff of the same university as a lecturer. 
In 1992 he earned the Docteur de l\emph{\'Ecole Nationale Sup\'erieure des
T\'el\'ecommunications} degree, in Paris. He is currently with the Statistics Departement of UFPE.
Dr. de Oliveira was appointed as honored professor by fivety electrical engineering
undergratuate classes and chosen as the godfather of sixteen
engineering graduation. His publications are available at arXiv. 
Research current interests include: applied statistics, communications theory,
applied information theory, signal analysis, and wavelets.
\end{IEEEbiography}
\end{minipage}

\begin{figure}[t]
{\includegraphics[width=0.35\textwidth]{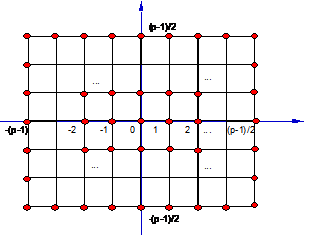}}
\caption{Two-dimensional constituent constellation.}
\label{fig:constellation}
\end{figure}

\end{document}